\documentclass[letterpaper,10pt,conference]{ieeeconf}
\IEEEoverridecommandlockouts
\usepackage{bbm, ifthen,xcolor}
\usepackage[utf8]{inputenc}
\usepackage{epsfig,epstopdf,color}
\usepackage{graphics,graphicx,amssymb,amsmath,verbatim,subfigure}
\usepackage{amsbsy,mathrsfs,amsfonts,amssymb,xcolor}
\usepackage{multirow,enumerate,cite,flushend}

\newcommand*{\QEDB}{\null\nobreak\hfill\ensuremath{\square}}%
\newcommand*{\R}{\mathbb{R}}%
\usepackage{mathtools}
\graphicspath{{./Figures/}}
\DeclareMathOperator {\He}{He}
\DeclareMathOperator {\dom}{dom}
\DeclareMathOperator {\rk}{rank}

\newtheorem{myremark}{Remark}
\newtheorem{myassumptions}{Assumption}
\newtheorem{mytheorem}{Theorem}
\newtheorem{mydefinition}{Definition}
\newtheorem{myproposition}{Proposition}
\newtheorem{property}{Property}
\newtheorem{problem}{Problem}
\renewcommand{\Re}{\mathbb{R}}
\renewcommand{\epsilon}{\in}
\def\T {\text{T}}

\allowdisplaybreaks

\begin{document}

\title{Output Regulation of Linear Aperiodic Sampled-Data Systems}
\author{Himadri Basu, Francesco Ferrante, Se Young Yoon\thanks{H. Basu is with the Department of Electrical and Biomedical Engineering, University of Vermont, Burlington, VT 05452, USA. Email: himadri.basu@uvm.edu, Francesco Ferrante is with the Department of Engineering, University of Perugia Italy. Email: francesco.ferrante@unipg.it, S. Y. Yoon is with Department of Electrical and Computer Engineering, University of New Hampshire, Durham, NH 03824, USA. Email: SeYoung.Yoon@unh.edu}\thanks{Research by Francesco Ferrante is funded in part by ANR via project HANDY, number ANR-18-CE40-0010.}}

\maketitle
\begin{abstract}
This paper deals with the output regulation problem of a linear time-invariant system in the presence of sporadically available measurement streams. A regulator with a continuous intersample injection term is proposed, where the intersample injection is provided by a linear dynamical system and the state of which is reset with the arrival of every new measurement updates. The resulting system is augmented with a timer triggering an instantaneous update of the new measurement and the overall system is then analyzed in a hybrid system framework. With the Lyapunov based stability analysis, we offer sufficient conditions to ensure the objectives of the output regulation problem are achieved under intermittency of the measurement streams. Then, from the solution to linear matrix inequalities, a numerically tractable regulator design procedure is presented. Finally, with the help of an illustrative example, the effectiveness of the theoretical results are validated.   
\end{abstract}


\section{Introduction}
The objective of an output regulation problem is to control a given output of the plant to asymptotically track a prescribed reference trajectory and rejecting asymptotically undesired disturbances, both of which are generated by an exosystem, while keeping all the trajectories of the system bounded \cite{14,Astolfi5}. In contrast to \cite{14}, where the measured output of the plant was assumed to be continuous, in this work we consider that the output measurement streams are only available sporadically. Owing to the intermittent availability of the measured plant output, the classical output regulation theory based on the internal model principle in \cite{8} is not applicable.  

Output regulation problem for linear networked control systems with measurement intermittency has been addressed in the works of \cite{Astolfi5}, where the impulsive updates of the latest measurements from the plant are subjected to a zero-order holding device. With the assumption that the impulsive new measurement updates are held constant until the measurement arrives, the results of \cite{Astolfi5} are then extended to the case of minimum phase nonlinear systems in \cite{Astolfi6}. Output regulation problem with periodically sampled measurement updates for linear systems are studied by the authors of \cite{Lawrence9}, \cite{Fujioka10}. The works in \cite{Lawrence9, Fujioka10, Antunes11} do not take into account the intersampling behavior and phenomenon like uncertain time-varying transmission and scheduling are neglected. While these issues are addressed in the context of output regulation problems for linear networked control systems in \cite{Astolfi5}, the proposed continuous-discrete regulator therein keeps the received plant measurement constant in between the sampling times, which is a restrictive requirement.

\subsection{Contribution}
In this paper, we study the output regulation problem for LTI systems with aperiodically sampled measurements. The main contribution of the paper consists of showing how the ``pre-processing'' and ``post-processing'' architectures in \cite{Wang9} can be adapted by including suitable ``hybrid extensions'' to assure asymptotic output regulation in the presence of intermittent output measurements. For the two architectures, sufficient conditions in the form of matrix inequalities are provided. The two  architectures are compared in terms of design complexity. In particular, we demonstrate that by opting for a ``post-processing paradigm'', the proposed design conditions are easier to handle from numerical standpoint as opposed to its counterpart, thereby offering greater design simplicity.


The remainder of the paper is organized in the following manner. First, in Section~\ref{sec:1} we briefly present some preliminaries on classical output regulation problem, and hybrid systems theory. The construction of the hybrid linear regulator based on pre-processing of internal model along with the hybrid modeling of overall closed loop system is proposed in Section~\ref{sec:2}. The sufficient conditions concerning the stability of the overall closed loop system with intermittent measurements are given in Section~\ref{sec:3}. The solution to the output regulation problem with the post-processing architecture is offered in Section \ref{sec:7}. A numerical example illustrating the effectiveness of the proposed solution by both approaches is given in  respective sections. Discussion on results and some concluding remarks are provided in Section \ref{sec:6}.
\subsection{Notation}
We will now introduce some notations which will be used throughout the text. The set $\mathbb{N}_{>0}$ is the set of all strictly positive integers and $\mathbb{N}=\mathbb{N}_{>0}\cup\{0\}$. The set of $\mathbb{R}_{>0}$\ (or $\mathbb{R}_{>=0}$) represent the set of positive (or non-negative) real numbers. $\mathbb{R}^{n\times{m}}$ represent the set of all real matrices with order $n\times{m}$. $I$ and $0$ are respectively the identity and null matrices of appropriate dimensions. A square, symmetric matrix $A=A^{\T}$ with $A\geq{0}$\ (or $A>0$) implies that the matrix $A$ is a positive semi-definite (or a positive definite) and equivalently $-A$ is a negative semi-definite (or negative definite) matrix. Given square matrices $A_1,\ A_2,\ ,\cdots,\ A_N$ of compatible dimensions, $A=\text{blk diag}(A_1,A_2,\cdots,A_N)$ denotes a block diagonal matrix with the $i^{\text{th}}$ diagonal element being $A_i$ and $\text{col}(A_1,A_2,\cdots,A_N)=\begin{bmatrix}A^{\T}_1,A^{\T}_2,\cdots,A^{\T}_{N}\end{bmatrix}^{\T}$. For a square matrix $A$, $\sigma(A)$ is the spectrum of $A$, ${\He}(A)=A+A^{\T}$, $\det(A)$ is the determinant and product of all eigenvalues of $A$. For $x,y\in\Re^{N}$, $\|x\|$ denotes the Euclidean norm of vector $x$, $\text{col}(x,y)=[x^{\T},y^{\T}]^{\T}$, $\langle{x,y}\rangle$ is the standard inner product. Given a vector $x\in\Re^{N}$ and a nonempty set $\mathcal{A}\subset\Re^{N}$, the distance of $x$ to $\mathcal{A}$ is defined as $|x|_{\mathcal{A}}=\inf_{y\in\mathcal{A}}\|x-y\|$. 
\subsection{Preliminaries on hybrid Systems}\label{subsec:hy}
We consider a hybrid system with state $x\in\mathbb{R}^{n}$ of the form
\begin{equation}\label{eq:hy}
\mathcal{H}\begin{cases}
\dot{x}\ =\ f(x),~ x\in C\\
x^{+}\ \in\ G(x),~ x\in D,   
\end{cases}
\end{equation}
where the shorthand notation $\mathcal{H}=\{C,f,D,G\}$ comprises of the flow set $C$, flow map $f$, jump set $D$, and jump map $G$. A set $E\subset\mathbb{R}_{>=0}\times \mathbb{N}$ is said to be a hybrid time domain if it is the union of finite or infinite sequence of intervals $[t_j,t_{j+1})\times\{j\}$. A function $\phi:\text{dom}\ \phi\mapsto\mathbb{R}^{n}$ is a hybrid arc if $\text{dom}\ \phi$ is a hybrid time domain with $\phi(\cdot\ ,\ j)$ being locally absolutely continuous for each $j$. A solution to $\mathcal{H}$ is said to be \textit{complete} if its domain is unbounded and \textit{maximal}, and if it is not the truncation of another solution \cite{Ferrante12}. Given a hybrid system $\mathcal{H}$ in \eqref{eq:hy}, we say that $\mathcal{H}$  satisfies hybrid basic conditions \cite{Goebel13,Ferrante12}, if $C$ and $D$ are closed in $\mathbb{R}^{n}$, $f:\mathbb{R}^{n}\to\mathbb{R}^{n}$ is continuous and $G\rightrightarrows\mathbb{R}^{n}$ is locally bounded, nonempty, outer semicontinuous relatively in $D$. 

The following notion of global exponential stability is used in the paper.
\begin{mydefinition}[\cite{teel2012lyapunov}]
Given the set $\mathcal{A}$ be closed. The set $\mathcal{A}$ is said to be globally exponentially stable (GES) for $\mathcal{H}$ if there exists strictly positive real numbers $\lambda$ and $k$ such that for any initial condition every maximal solution $\phi$ to $\mathcal{H}$ is complete and satisfies for all $(t,j)\in\text{dom}\ \phi$
\begin{equation}\label{eq:16}
|\phi(t,j)|_{\mathcal{A}}\leq k\ e^{-\lambda(t+j)}|\phi(0,0)|_{\mathcal{A}}.
\end{equation}
\end{mydefinition}
 
\section{Problem Formulation}\label{sec:1}
Consider a linear time-invariant plant of the form
\begin{equation}
\begin{aligned}\label{eq:plant}
\mathcal{P}\begin{cases}
\dot{x}_{p} &= A_{p}x_p+B_{p}u+E_{p}w,\\
y_{p} &= C_{p}x_p,\\
e_{p} &= y_{p}- y_{w},\\
y_w &= F_{p}w,
\end{cases}
\end{aligned}
\end{equation}
with $x_{p}\in\Re^{n_p},~u\in\Re^{m_p},~y_p,~e_p\in\Re^{p}$ being respectively the state, control law to be designed, measured and regulated output of the plant which we aim to regulate to zero. The exogenous signal $w\in\Re^{q}$ is generated by an exosystem of the form 
\begin{equation}\label{eq:2}
\dot{w}=Sw,
\end{equation}
where the exosystem matrix $S$ is assumed to be neutrally stable, \emph{i.e.} $S$ has all eigenvalues on the imaginary axis. While $S$ is perfectly known, the exosystem state $w$ in \eqref{eq:2} not directly available for feedback design. The matrices $A_p, B_p, E_p, C_p$ and $F_p$ in \eqref{eq:plant} are constant matrices of appropriate dimensions and such that the pair $(A_p,C_p)$ is detectable. The output $y_p$ is available only at some isolated time instances $t_k, k\in{\mathbb{N}}_{>0}$, not known a priori. We assume that the sequence $\{t_k\}_{k=1}^{\infty}$ are strictly increasing, $t_k\to\infty$ as $k\to\infty$ and there exist two positive scalars $T_1$ and $T_2$ which uniformly bounds the consecutive intersampling intervals $[t_k,t_{k+1}]$ of $\{t_k\}$ as follows
\begin{equation}\label{eq:3}
0\leq{t_1}\leq{T_2},~T_1\leq{t_{k+1}-t_k}\leq{T_2},~\forall{k}\in\mathbb{N}_{>0}.
\end{equation} 
As noted in \cite{Ferrante3}, the strictly positive lower bound $T_1$ prevents the existence of accumulation points in the sequence $\{t_k\}_{k=1}^{\infty}$ and thus avoids zero behaviors. On the other hand, $T_2$ defines the maximum allowable transfer time (MATI) \cite{Postoyan7}.
\subsection{Preliminaries on Linear Output Regulation Problem}\label{sec:prl}
We now define the objectives of the output regulation problem. For system \eqref{eq:plant}, the output regulation problem is said to be solved if the following two conditions are satisfied.
\begin{enumerate}
\item When $w=0$, all the trajectories of the closed loop system exponentially converge to zero, \emph{i.e.} the origin of the unperturbed closed loop system $(w=0)$ is exponentially stable.
\item When $w\neq{0}$, the trajectories of the closed loop system \eqref{eq:plant} are internally stable and the regulated output signal $e_p(t)$ exponentially converges to zero, \emph{i.e.} $\lim_{t\to\infty}e_p(t)=0$.
\end{enumerate}

We now consider the following two assumptions which are required to guarantee the solvability of the classical output regulation problem \cite{14,8}.
\begin{myassumptions}\label{as:1}
The matrix pair $(A_p,B_p)$ is stabilizable and $(A_p,C_p)$ is detectable.\hfill $\diamond$
\end{myassumptions}

\begin{myassumptions}\label{as:2}
The matrix $\begin{bmatrix}A_p-\lambda{I} & B_p\\C_p & 0\end{bmatrix},\ \lambda\in\sigma(S)$ is of full rank or equivalently there exists a unique solution pair $(X_p,R)$ to the following linear regulator equation
\begin{equation}\label{eq:4}
\begin{aligned}
X_pS &=A_pX_p+B_p{R}+E_p,\\
0 &=C_pX_p-F_p,
\end{aligned}
\end{equation} 
where the matrix $X_p$ uniquely defines the steady state $x_p=X_pw$ on which the regulated output $e_{p}=0$. Additionally, the steady state input $u=Rw$ renders the given manifold $x_p=X_pw$ positively invariant. \hfill $\diamond$
\end{myassumptions}

As noted in \cite{14}, under Assumptions \ref{as:1}, \ref{as:2}, the output regulation problem for the plant \eqref{eq:plant} is solvable by the dynamic error feedback control of the form 
\begin{equation}
\begin{aligned}\label{eq:5}
u &={K}z,~
\dot{z} &=\mathcal{G}_{1}z+\mathcal{G}_{2}e_p,
\end{aligned}
\end{equation}
where $z\in\Re^{n_z}$ is the regulator state to be specified later and the constant controller gain matrices $\mathcal{G}_1$ and $\mathcal{G}_{2}$ of appropriate dimensions are defined as
\begin{equation}\label{eq:6}
\begin{aligned}
\mathcal{G}_{1} &=T\begin{bmatrix}S_1 & S_2\\S_3 & G_1\end{bmatrix}{T^{-1}},~\mathcal{G}_{2}=T\begin{bmatrix}S_4\\G_2\end{bmatrix},\\
G_1 &=\text{blk diag}(\beta_1,\beta_2,\cdots,\beta_p),\\
{G}_2 &=\text{blk diag}(\gamma_1,\gamma_2,\cdots,\gamma_p),
\end{aligned}
\end{equation}
with $\sigma(G_1)=\sigma(S)$, and $\beta_i,\ \gamma_i$ respectively being a constant square matrix and a column vector of dimension $d_i\in{\mathbb{N}}_{>0}$ such that the pair $(\beta_i,\gamma_i)$ is controllable and the matrix $\hat{A} =\begin{bmatrix}A & BK\\\mathcal{G}_{2}C& \mathcal{G}_1\end{bmatrix}$ is Hurwitz. The matrices $S_1,\ S_2,\ S_3,\ S_4$ are arbitrary constant matrices of appropriate dimensions and $T\in\Re^{n_z\times{n_z}}$ is any nonsingular matrix. Therefore, when the signals $y_p$ and $y_w$ are continuously measured, the requirements for the solvability of the output regulation problem are said to be satisfied with Assumptions \ref{as:1} and \ref{as:2}.

\section{Solution Outline by Pre-processing architecture of Internal Model}\label{sec:2}
\subsection{Proposed Controller} Since the output of the plant is available sporadically, we propose a control scheme, depicted in Fig. \ref{fig.1}, constituted by a preprocessing of an internal model $\mathcal{G}$ of the exosystem \cite{Bin14}, stabilizing controller $\mathcal{K}$, and a holding device $\mathcal{J}$. In the proposed control scheme, the plant $\mathcal{P}$ along with the internal model of the exosystem $\mathcal{G}$, viewed together as an extended continuous plant $\hat{\mathcal{P}}$ is stabilized by a dynamic controller $\mathcal{K}$ which relies on the continuous regulated error signal $\theta$ generated by the holding device $\mathcal{J}$. The hold device $\mathcal{J}$ receives the intermittent regulated error signal $e(t_k)$ available from the output of the sampler $\mathcal{S}$ at every nonuniform time instants $t_k$. 

With a little abuse of notation from \eqref{eq:5}, the continuous-time internal model controller $\mathcal{G}$ and its input to the plant $\mathcal{P}$ is given as follows
\begin{equation}\label{eq:int_model}
\mathcal{G}\left\{
\begin{aligned}
&\dot{z}=G_{1}z+G_{2}v\\
&u= Kz
\end{aligned}\right.
\end{equation} 
where $K\in\mathbb{R}^{m_{p}\times{n}_{z}}$ is a controller gain matrix of the extended plant $\hat{\mathcal{P}}$, the internal model controller pair $(G_1,G_2)$ are defined in \eqref{eq:6}, and the continuous-time signal $v(t)\in\mathbb{R}^{n_v}$ is the output of the stabilizer, defined next. 
\begin{equation}\label{eq:stab}
\mathcal{K}
\begin{cases}
\dot{x}_{c}= A_c{x_c} + B_c{\theta},\\
v \ = C_c{x_{c}} + D_c{\theta},
\end{cases}
\end{equation}
where $x_{c}\in\mathbb{R}^{n_p+n_z}$ is the stabilizer state and $\theta\in\mathbb{R}^{p}$ is the state of the holding device $\mathcal{J}$. From the controller state $x_{c}(t)$ and last received measurement of the regulated output $e_{p}(t)$, the holding device $\mathcal{J}$ generates an intersample signal to feed the stabilizer $\mathcal{K}$. For all $k\in\mathbb{N}_{>0}$ The dynamics of the holding device $\mathcal{J}$ is given as follows
\begin{equation}\label{eq:theta}
\mathcal{J}
\begin{cases}
\dot{\theta}(t)= H\theta + Ex_c,\ \text{if}\,\,t\neq\ t_k,\\
\theta(t^+)= e_p(t),\hspace{4ex}\ \text{if}\,\,t=t_k.
\end{cases}
\end{equation} 
The arrival of new measurements instantaneously updates $\theta(t)$ to $e_p(t)$ and in between updates the holding state $\theta$ evolves according to the continuous dynamics of \eqref{eq:theta}. 

In the next section, we present the hybrid modeling of the overall closed loop system in the presence of sporadic measurements of the regulated output $e_{p}(t)$. But, to simplify our analysis in the modeling stage, we first transform coordinates of the plant state $x_p$, internal model state $z$, and holding state $\theta$ as follows. 
\begin{figure}
\centering
\includegraphics[width=0.48\textwidth]{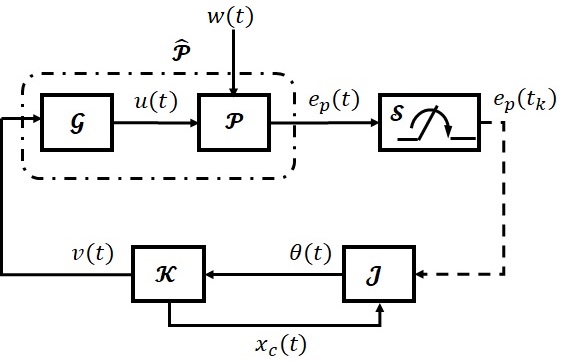}
\caption{{\label{fig.1}} Schematic representation of the closed loop system with sporadic measurements under ``pre-processing'' architecture. Continuous-time signals are marked with solid arrows, while the sporadic measurements are with dashed arrows.}\vspace{-1ex}
\end{figure}
\begin{equation}\label{eq:coord}
\tilde{x}_p=x_p-X_p{w},\ \tilde{z}=z-Zw,\ \tilde{\theta}=\theta-e_p,
\end{equation}
where $X_p$ is a solution to the linear regulator equation \eqref{eq:4}, and $Z\in\mathbb{R}^{q\times{n_z}}$ is a transformation matrix which by virtue of internal model principle \cite{14} satisfies
\begin{equation}\label{eq:Z}
ZS=G_1{Z},\ KZ=R
\end{equation}
with $G_1$ and $K$ given in \eqref{eq:int_model}.
\subsection{Hybrid modeling}
The hybrid closed-loop system with state $\tilde{x}=\text{col}(\tilde{x}_p,\tilde{z},x_c)$\ $\in\mathbb{R}^{2(n_p+n_z)}$ and jumps in $\theta(t)$, depicted in Fig. \ref{fig.1}, is described as follows
\begin{equation}\label{eq:cl}
\begin{aligned}
& \dot{\tilde{x}} = \mathbb{A}\ \tilde{x} + \mathbb{B}\ \tilde{\theta},~\dot{\tilde{\theta}} = H\tilde{\theta} + \mathbb{J}\ \tilde{x},~t\neq{t_k},\\
& \tilde{x}(t^+) =\tilde{x},~\tilde{\theta}(t^+)=0,\hspace{10ex}t=t_k,
\end{aligned}
\end{equation} 
where 
\begin{align*}
\mathbb{A} &=\begin{bmatrix}A+BD_c{C} & B{C_c}\\B_c{C} & A_c\end{bmatrix}, \mathbb{B}=\begin{bmatrix}B{D_c}\\B_c\end{bmatrix},\nonumber\\ 
\mathbb{J} &=\begin{bmatrix}HC-C{A} & E\end{bmatrix},\ A=\begin{bmatrix}A_p & B_p{K}\\0 & G_1\end{bmatrix},\\
B &=\begin{bmatrix}0\\G_2\end{bmatrix}, C=\begin{bmatrix}C_p & 0\end{bmatrix}.
\end{align*}

Similar to \cite{Ferrante3}, we now introduce a timer variable $\tau$ which keeps track of the duration of flows and triggers a jump when certain condition is violated. Therefore, from \cite{Ferrante3,Ferrante12}, $\tau$ is made to decrease as ordinary time increases satisfying \eqref{eq:3} and it resets to any point in $[T_1,T_2]$ when $\tau$ reaches $0$. The overall closed loop system composed of the states $\tilde{x},\ \tilde{\theta},\ \tau$ can then be represented by the following hybrid system 
\begin{equation}\label{eq:14}
\mathcal{H}_{cl}\begin{cases}
\dot{\tilde{\mathbf{x}}} &=f(\tilde{\mathbf{x}}),~\tilde{\mathbf{x}}\in {C},\\
\tilde{\mathbf{x}}^{+} &=G(\tilde{\mathbf{x}}),~\tilde{\mathbf{x}}\in{D},
\end{cases}
\end{equation}

\subsection{Problem Statement}\label{sub:prob_state}
To solve the output regulation problem, we introduce the compact set
\begin{equation}
    \label{eq:A}
    \mathcal{A}\coloneqq \{0\}\times\{0\}\times[0, T_2]\subset\mathbb{R}^{n_p+n_z+1}
\end{equation}
and the design the controller so that $\mathcal{A}$ is globally exponentially stable for hybrid system \eqref{eq:14}.

\begin{problem}
Given the extended plant model $\hat{\mathcal{P}}$ composed of $\mathcal{P}$ in \eqref{eq:plant} and internal model $\mathcal{G}$ in \eqref{eq:int_model}, design control $\Delta_{\mathcal{K}}$ and hold parameter $\Delta_{\mathcal{J}}$
\begin{equation}
\Delta_{\mathcal{K}}=\left[\begin{array}{c|c}A_c & B_c \\  \hline C_c & D_c\end{array}\right],\ \Delta_{\mathcal{J}}=\left[\begin{array}{c|c} H & E\end{array}\right]
\end{equation}
such that the set $\mathcal{A}$ defined in \eqref{eq:A} is GES for the hybrid system $\mathcal{H}_{cl}$.\hfill$\diamond$
\label{prob:Design}
\end{problem}
\section{Main Results}\label{sec:3}
In this section, we provide sufficient stability conditions to solve Problem~\ref{prob:Design}. To this end, following \cite{Ferrante12}, we introduce the following property, whose role is clarified later in Theorems~\ref{th:1} and \ref{th:2}. 
 we consider the fo Lyapunov function $V(\tilde{\mathbf{x}})$ for $\mathcal{H}_{cl}$ in the form as $\mathbf{V}(\tilde{\mathbf{x}})=W_1(\tilde{x})+W_{2}(\tilde{\theta},\tau)$, where $W_{1}(\tilde{x})=\tilde{x}^{\T}P_{1}\tilde{x}, W_{2}(\tilde{\theta},\tau)=e^{\delta\tau}{\tilde{\theta}}^{\T}P_2\tilde{\theta}$ with $P_1\in\mathbb{R}^{2(n_p+n_z)\times{2(n_p+n_z)}},\ P_2\in\mathbb{R}^{p\times{p}}>0$ and $\tau\in [0,T_2]$. Take $\chi_1=\min(\lambda_{\min}(P_1),\lambda_{\min}(P_2))$, $\chi_2=\max(\lambda_{\max}(P_1),e^{\delta{T_2}}\lambda_{\max}(P_2))$, then 
 \begin{equation}\label{eq:chi}
 \chi_{1}|\tilde{\mathbf{x}}|_{\mathcal{A}}^{2}\leq{\mathbf{V}(\tilde{\mathbf{x}})}\leq\chi_{2}|\tilde{\mathbf{x}}|_{\mathcal{A}}^{2}.
 \end{equation}
\begin{property} 
Consider positive definite continuously differentiable functions $W_{1}(\tilde{x})$ and $W_2(\tilde{\theta},\tau)$ with flow maps in \eqref{eq:14}. There exist positive definite functions $\rho_1\in\mathbb{R}^{2(n_p+n_z)}\to\mathbb{R}$, $\rho_{2}\in\mathbb{R}^{p}\to\mathbb{R}$, $\omega_{1}\in\mathbb{R}^{p}\to\mathbb{R}$, $\omega_{2}\in\mathbb{R}^{2(n_p+n_z)}\to\mathbb{R}$, positive scalars $k_1,~k_2$ such that $\forall\ (\tilde{x},v_1)\in\mathbb{R}^{2(n_p+n_z)+p}$, and $(\tilde{\theta},\tau,v_2)\in\mathbb{R}^{p}\times{[0,T_2]}\times\mathbb{R}^{2(n_p+n_z)}$ we have
\begin{align}
\langle\ {\nabla W_{1}(\tilde{x}),\mathbb{A}\tilde{x}+\mathbb{B}v_1}\ \rangle &\leq{-\rho_{1}}(\tilde{x})+\rho_{2}(v_{1}),\label{eq:19}\\
\langle\ {\nabla W_{2}({\tilde{\theta}},\tau), (H\tilde{\theta}+\mathbb{J}{v}_{2}}, -1)\ \rangle &\leq{-\omega_{1}}(\tilde{\theta})+\omega_{2}(v_{2}),\label{eq:20}\\
-\rho_{1}(\tilde{x})+\omega_{2}(\tilde{x}) & \leq{-k_1}\|\tilde{x}\|^2,\label{eq:21}\\
-\omega_{1}(\tilde{\theta})+\rho_{2}(\tilde{\theta}) & \leq{-k_2}\|\tilde{\theta}\|^2\label{eq:22}
\end{align}
\hfill$\diamond$
\end{property}
\begin{mytheorem}[\cite{Ferrante12}]\label{th:1}
Let \textit{Property 1} hold. Then the set $\mathcal{A}$ in \eqref{eq:A} is globally exponentially for the hybrid system $\mathcal{H}_{cl}$. \QEDB
\end{mytheorem}
 \begin{proof}
 {From \eqref{eq:chi}, we observe that the Lyapunov function $\mathbf{V}(\tilde{\mathbf{x}})$ is bounded between two monotonically increasing functions. Next, for each $\tilde{\mathbf{x}}\in{D}$ and a scalar $v_3\in [0,T_2]$ with $\bar{g}(\tilde{\mathbf{x}})=(\tilde{x},0,v_3)\in G(\tilde{\mathbf{x}})$ we have
 \begin{equation}\label{eq:23}
 \mathbf{V}(\bar{g})\!-\!\mathbf{V}(\tilde{\mathbf{x}})\!=\!W_2(0,\!v_3)-W_2(\tilde{\theta},\!0)\!\leq\!\!{-\lambda_{\max}(\!P_2)}\|\tilde{\theta}\|^2.
 \end{equation}
 On the other hand, by evaluating $\dot{\mathbf{V}}=\langle {\nabla \mathbf{V}(\tilde{\mathbf{x}}), f(\tilde{\mathbf{x}})} \rangle$ along flow directions in \eqref{eq:14} and by virtue of equations \eqref{eq:19} - \eqref{eq:22}, we obtain
 \begin{align}
 & \langle {\nabla \mathbf{V}(\tilde{\mathbf{x}}), f(\tilde{\mathbf{x}})} \rangle \!=\! \langle {\nabla W_1,\mathbb{A}\tilde{x}\!+\!\mathbb{B}\tilde{\theta}} \rangle \!+\! \langle {\nabla {W}_2, \mathbb{J}\tilde{x}\!+\!H\tilde{\theta}} \rangle\nonumber\\
 & \leq -\rho_{1}(\tilde{x})+\rho_{2}(\tilde{\theta})-\omega_{1}(\tilde{\theta})+\omega_{2}(\tilde{x})\nonumber\\
 & \leq -k_1(\|\tilde{x}\|^2)-k_2(\|\tilde{\theta}\|^2)\leq -\chi_3\ |\tilde{\mathbf{x}}|^2_{\mathcal{A}},\label{eq:24}
 \end{align}
 where $\chi_3=\min(k_1,k_2)$. From \eqref{eq:chi}, $\dot{\mathbf{V}}$ in \eqref{eq:24} yields $\dot{\mathbf{V}}(\tilde{\mathbf{x}})\leq -\dfrac{\chi_3}{\chi_2}\ \mathbf{V}(\tilde{\mathbf{x}})$ for all $\tilde{x}\in C$, and therefore, thanks to \eqref{eq:23}, $\forall\ (t,j)\in \text{dom}\ \phi_{cl}$, 
 \begin{equation}
 \mathbf{V}(\phi_{cl}(t,j)) \leq e^{-\dfrac{\chi_3}{\chi_2}t} \mathbf{V}(\phi_{cl}(0,0)),\nonumber
 \end{equation}
 or equivalently,
 \begin{align}
 |\phi_{cl}(t,j)|_{\mathcal{A}} &\leq \sqrt{\dfrac{\chi_{2}}{\chi_{1}}} \ e^{-\dfrac{\chi_3}{2\chi_2}t}|\phi_{cl}(0,0)|_{\mathcal{A}}.
 \end{align} 
 As a result, the conditions of global exponential stability \eqref{eq:16} are satisfied with $\lambda\in \bigg(0,\dfrac{\chi_{3}T_1}{2\chi_{2}(1+T_1)}\bigg]$ and $k_{cl}=\sqrt{\dfrac{\chi_2}{\chi_1}}e^{\omega},\ \omega\geq\lambda$ and hence the set $\mathcal{A}$ is globally exponentially stable with respect to $\mathcal{H}_{cl}$.}
 \end{proof}

\begin{mytheorem}
\label{th:2}
If there exist symmetric positive definite matrices $P_3,P_4\in\mathbb{R}^{2(n_p+n_z)\times{2(n_p+n_z)}}$, $P_5, P_6\in\mathbb{R}^{p\times{p}}$, and matrices $A_c\in\mathbb{R}^{(n_p+n_z)\times{(n_p+n_z)}}$, $B_c\in\mathbb{R}^{(n_p+n_z)\times{p}}$, $C_c\in\mathbb{R}^{n_v\times{(n_p+n_z)}}$, $D_c\in\mathbb{R}^{n_v\times{p}}$, $H\in\mathbb{R}^{p\times{p}}$, $E\in\mathbb{R}^{p\times{(n_p+n_z)}}$ be such that
\begin{align}
 & P_3-P_4  \prec 0,\label{eq:26}\\
 & P_5-P_6  \prec 0,\label{eq:27}\\
 & \mathcal{M}_{1} = \begin{bmatrix}{\He}(P_1\mathbb{A})+P_4 & P_1\mathbb{B}\\
\mathbb{B}^{\T}P_1 & -P_5
\end{bmatrix} \preceq {0},\label{eq:28}\\
 & \mathcal{M}_2(0)\preceq{0},\ \mathcal{M}_{2}(T_2)\preceq {0},\label{eq:29} 
\end{align}
where $\mathcal{M}_{2}(\tau)=\begin{bmatrix}e^{\delta\tau}\Big[{\He}(P_2H)-\delta{P_2}\Big]+P_6 & e^{\delta\tau}P_2\mathbb{J}\\
e^{\delta\tau}\mathbb{J}^{\T}P_2 & -P_3\end{bmatrix}$ with $\tau\in[0,T_2]$, then \textit{Property 1} holds. \QEDB
\end{mytheorem}
\begin{proof}
Define $\rho_1(\tilde{x})=\tilde{x}^{T}P_4\tilde{x}$, $\rho_{2}(\tilde{\theta})=\tilde{\theta}^{\T}P_5\tilde{\theta}$, $\omega_{1}(\tilde{\theta})=\tilde{\theta}^{\T}P_6\tilde{\theta}$, $\omega_{2}(\tilde{x})=\tilde{x}^{\T}P_3\tilde{x}$. Now we evaluate $\dot{W}_{1}$ as 
\begin{align}
\dot{W}_{1}(\tilde{x}) &=\langle \nabla W_{1}(\tilde{x}),\mathbb{A}\tilde{x}+\mathbb{B}\tilde{\theta} \rangle=\dot{\tilde{x}}^{\T}P_1\tilde{x}+\tilde{x}^{\T}P_1\dot{\tilde{x}}\nonumber\\
& =\tilde{x}^{\T}{\He}(P_1\mathbb{A})\tilde{x}+\tilde{\theta}^{\T}\mathbb{B}^{\T}P_1\tilde{x}+\tilde{x}^{\T}P_1\mathbb{B}\ \tilde{\theta}.
\end{align}
By virtue of \eqref{eq:28}, we thus obtain $\Omega_{1}(\tilde{x},\tilde{\theta})=\dot{W}_{1}+\tilde{x}^{\T}P_4\tilde{x}-\tilde{\theta}^{\T}P_5\tilde{\theta}=\begin{bmatrix}\tilde{x}^{\T} & \tilde{\theta}^{\T}\end{bmatrix}\mathcal{M}_{1}\begin{bmatrix}\tilde{x}\\\tilde{\theta}\end{bmatrix}\leq{0}$, which as a consequence yields \eqref{eq:19}. On the other hand, 
\begin{align}
&\dot{W}_{2}=e^{\delta\tau}\dot{\tilde{\theta}}^{\T}P_2\tilde{\theta}+e^{\delta\tau}{\tilde{\theta}}^{\T}P_2\dot{\tilde{\theta}}-\delta{e^{\delta\tau}}\tilde{\theta}^{\T}P_2\tilde{\theta}\nonumber\\
&=e^{\delta\tau}\Big[\tilde{\theta}^{\T}\Big({\He}(P_2{H})\!-\!\delta{P_2})\tilde{\theta}\!+\!\tilde{x}^{\T}\mathbb{J}^{\T}P_2\tilde{\theta}\!+\!\tilde{\theta}^{\T}P_2\mathbb{J}\tilde{x}\Big].\label{eq:31}
\end{align}
Then, $\Omega_{2}(\tilde{x},\tilde{\theta})\!=\!\dot{W}_{2}+\tilde{\theta}^{\T}P_6\tilde{\theta}-\tilde{x}^{\T}P_3\tilde{x}=\begin{bmatrix}\tilde{x}^{\T}\! & \! \tilde{\theta}^{\T}\end{bmatrix}\!\mathcal{M}_{2}(\tau)\!\begin{bmatrix}\tilde{x}\\\tilde{\theta}\end{bmatrix}$, which is a convex expression with respect to each value of $\tau\in [0,T_2]$ and therefore, for each $\tau$, there exists $\alpha(\tau)$ such that $\mathcal{M}_{2}(\tau)=\alpha(\tau)\mathcal{M}_{2}(0)+(1-\alpha(\tau))\mathcal{M}_{2}(T_2)$. By virtue of $\mathcal{M}_{2}(0),\ \mathcal{M}_{2}(T_2)\leq{0}$, we thus obtain $\mathcal{M}_{2}(\tau)\leq {0}$ and consequently $\Omega_{2}(\tilde{x},\tilde{\theta})\leq{0}$, which in turn yields \eqref{eq:20}.

Since $\dot{\Omega}_{1}(\tilde{x},\tilde{\theta}),\ \dot{\Omega}_{2}(\tilde{\theta},\tau)\leq{0}$, we further obtain from Eqns. \eqref{eq:26} and \eqref{eq:27}
\begin{align}
\dot{\mathbf{V}} &=\dot{W}_{1}+\dot{W}_{2}\leq \tilde{x}^{\T}(P_3-P_{4})\tilde{x}+\tilde{\theta}^{\T}(P_5-P_6)\tilde{\theta}\nonumber\\
& \leq \lambda_{\max}(P_3-P_4)\|{\tilde{x}}\|^2+\lambda_{\max}(P_5-P_6)\|\tilde{\theta}\|^2\nonumber\\
& \leq -k_1\|\tilde{x}\|^2-k_2\|\tilde{\theta}\|^2\leq -\chi_3 |\tilde{\mathbf{x}}|^2_{\mathcal{A}}<0,
\end{align}
where $k_1=-\lambda_{\max}(P_3-P_4)$, $k_2=-\lambda_{\max}(P_5-P_6)$, and as a result, the last two conditions of \textit{Property 1} are evident.
\end{proof}

Theorem~\ref{th:2} provides sufficient conditions to guarantee the exponential stability of $\mathcal{H}_{cl}$ with respect to $\mathcal{A}$. However, these conditions in \eqref{eq:26} - \eqref{eq:29} can not be directly used for designing the decision variables $P_1,\ P_2,\ A_c,\ B_c,\ C_c,\ D_c$ and $\delta$. Therefore, some matrix manipulations are required to turn these conditions into an LMI feasibility problem. The proposed control solution to the output regulation problem by the pre-processing architecture is an extension of the results on exponential stabilization of LTI systems in the presence of aperiodic sampling, presented in \cite{Ferrante12}. 

\subsection{LMI based Regulator Design}\label{subsec:4}
In this section, we perform matrix and variable manipulations to turn conditions \eqref{eq:26} - \eqref{eq:29} into a tractable LMI based controller design procedure. First, we find the Schur complement of \eqref{eq:28} as
\begin{equation}\label{eq:33}
\bar{\mathcal{M}}_{1} =\begin{bmatrix}{\He}(P_1\mathbb{A}) & P_1\mathbb{B} & I\\
\mathbb{B}^{\T}P_1 & -P_5 & 0\\
I & 0 & -P_8
\end{bmatrix} \preceq{0}
\end{equation}
where $P_8=P^{-1}_{4}$. Therefore, \eqref{eq:26} now becomes
\begin{equation}\label{eq:34}
P_3 - P^{-1}_{8} \prec 0,
\end{equation}
which is not an LMI with respect to $P_8$. To transform this into an LMI, we need to find an upper bound of $P_3$ in \eqref{eq:34} in terms of $P_8$. From Lemma 1 of \cite{Ferrante12}, for all $\alpha\in\mathbb{R}$, $P^{-1}_{8}$ and $P_8$ are related by the following inequality
\begin{equation} \label{eq:35}
P^{-1}_{8}\succ 2\alpha{I}-\alpha^2P_8.
\end{equation}
Therefore, from \eqref{eq:35}, the matrix conditions in \eqref{eq:34} are met if we assume that the following LMI holds:
\begin{equation} \label{eq:36}
P_3-2\alpha{I}+\alpha^{2}P_8\prec 0,\ \alpha\in\mathbb{R}.
\end{equation}
Next, in \eqref{eq:33}, we observe that the nonlinear terms are associated with the decision variable $P_1$. Let us now characterize the structure of $P_1$ as
\begin{equation}
\begin{aligned} \label{eq:37}
P_1 &=\begin{bmatrix}X & U\\
U^{T} & W
\end{bmatrix},\ P^{-1}_1=\begin{bmatrix}Y & V\\V^{\T} & Z\end{bmatrix},
\end{aligned}
\end{equation}
where $\ X,Y,Z,W\in\mathbb{R}^{n_p+n_z}>0,\ {\det}(U)\neq{0},\ {\det}{V}\neq{0}$ for $U,V\in\mathbb{R}^{n_p+n_z}$. From \eqref{eq:37}, we obtain $W=-V^{-1}YU=-V^{-1}Y(I-XY)V^{-T}$ since $XY+UV^{\T}=I$, and therefore $P_1$ becomes
\begin{equation} \label{eq:38}
P_1=\left[\begin{array}{c|c}X & U \\ \hline \noalign{\vspace{0.7ex}} {U^{T}} & -V^{-1}(Y-YXY)V^{-T}\end{array}\right].
\end{equation}
Next, define a matrix $\Psi=\begin{bmatrix}Y & I\\V^{\T} & 0\end{bmatrix}$ which is nonsingular as ${\det}(V)\neq{0}$. Since $P_1>0$, we also obtain
\begin{equation} \label{eq:39}
\bar{P}_{1}=\Psi^{\T}P_1\Psi=\begin{bmatrix}Y & I\\I & X\end{bmatrix}>0,
\end{equation}
and then by congruence transformation on $\bar{\mathcal{M}}_{1}$, \eqref{eq:33} yields 
\begin{align}
\hat{\mathcal{M}}_{1} &=\text{blkdiag}(\Psi^{\T},I,I)\ \bar{\mathcal{M}}_{1}\ \text{blk diag}(\Psi,I,I)\nonumber\\
&= \begin{bmatrix}\Psi^{\T}{\He}(P_1\mathbb{A})\Psi & \Psi^{\T}P_1\mathbb{B} & \Psi^{\T}\\
\mathbb{B}^{\T}P_1\Psi & -P_5 & 0\\
\Psi & 0 & -P_8
\end{bmatrix}\nonumber\\
&={\He}\Bigg(\begin{bmatrix}\Psi^{\T}P_1\mathbb{A}\Psi & \Psi^{\T}P_1\mathbb{B} & 0\\
0 & -P_5 & 0\\
\Psi & 0 & -P_8\end{bmatrix}\Bigg)
\leq{0}, \label{eq:40con}
\end{align} 
which is still not an LMI. By performing the change of variables as in \cite{Ferrante12} we define $K_1\in\mathbb{R}^{n_v\times(n_p+n_z)}$, $K_2\in\mathbb{R}^{n_v\times{p}}$, $K_3\in\mathbb{R}^{(n_p+n_z)\times(n_p+n_z)}$, $K_4\in\mathbb{R}^{(n_p+n_z)\times{p}}$ such that
\begin{align}
& \begin{bmatrix}K1 & K_2\\K_3-XAY & K_4\end{bmatrix}=\!\!\begin{bmatrix}I & 0\\XB & U\end{bmatrix}\!\!\begin{bmatrix}D_c & C_c\\B_c & A_c\end{bmatrix}\!\!\begin{bmatrix}CY & I\\V^{\T} & 0\end{bmatrix}\nonumber\\
& \Delta_{\mathcal{K}}=\begin{bmatrix}U^{-1} & -U^{-1}XB\\0 & I\end{bmatrix}\Gamma_{\mathcal{K}}\begin{bmatrix}V^{-T} & 0\\-C{Y}V^{-T} & I\end{bmatrix}, \label{eq:41}
\end{align}
where $\Gamma_{\mathcal{K}}=\begin{bmatrix}K_3-XA{Y} & K_4\\K_1 & K_2\end{bmatrix}$. Then, by substituting the results of \eqref{eq:41} in \eqref{eq:40con} we obtain
\begin{align}
\hat{\mathcal{M}}_{1} &={\He}\Bigg(\begin{bmatrix}\Pi_{1} & \Pi_2 & 0\\0 & -P_5 & 0\\\Psi & 0 & -P_8
\end{bmatrix}\Bigg)\leq{0}, \label{eq:42}\\
\Pi_{1} &=\begin{bmatrix}A{Y}+BK_1 & A_0+B{K_2}C\\
K_3 & XA+K_{4}C\end{bmatrix},\Pi_{2}=\begin{bmatrix}BK_2\\K_4\end{bmatrix},\nonumber
\end{align}
which is an LMI with respect to $X, Y, V, K_i,~i=1,2,3,4$. Next, by defining $Z_1=P_2{H}$ and $Z_2=P_2{E}$, it can be easily shown that the equation \eqref{eq:29} can be turned into an LMI with respect to the decision variables $P_2,Z_{1},Z_2,P_6$. The hold gain $\Delta_{\mathcal{J}}$ thus becomes
\begin{equation} \label{eq:42.1}
\Delta_{\mathcal{J}}=\left[\begin{array}{c|c}P^{-1}_{2}Z_1 & P^{-1}_{2}E\end{array}\right].
\end{equation}
In what follows is a proposition with a set of LMIs giving sufficient conditions to solve \textit{Problem 1}. The proof of the results in Proposition \ref{prop.1} appear in \cite{Ferrante12}. 
\begin{myproposition} \label{prop.1}
Given a plant $\mathcal{P}$ in \eqref{eq:plant} with internal model controller $\mathcal{G}$ in \eqref{eq:int_model}, real scalars $\alpha,\delta>0$, $X,Y\in\mathbb{R}^{(n_p+n_z)\times(n_p+n_z)}>0$, $P_3,P_8\in\mathbb{R}^{2(n_p+n_z)\times{2}(n_p+n_z)}>0$, $P_5, P_6\in\mathbb{R}^{p\times{p}}>0$, $V\in\mathbb{R}^{(n_p+n_z)\times{(n_p+n_z)}}$ with ${\det} (V)\neq{0}$, $K_1\in\mathbb{R}^{n_v\times(n_p+n_z)}$, $K_2\in\mathbb{R}^{n_v\times{p}}$, $K_3\in\mathbb{R}^{(n_p+n_z)\times(n_p+n_z)}$, $K_4\in\mathbb{R}^{(n_p+n_z)\times{p}}$, $Z_{1}\in\mathbb{R}^{p\times{p}}, Z_{2}\in\mathbb{R}^{p\times(n_p+n_z)}$ such that
\begin{align}
\bar{P}_{1} &=\begin{bmatrix}Y & I\\I & X\end{bmatrix}>0,\label{eq:43}\\
\hat{\mathcal{M}}_{1} &={\He}\Bigg(\begin{bmatrix}\Pi_{1} & \Pi_2 & 0\\0 & -P_5 & 0\\\Psi & 0 & -P_8
\end{bmatrix}\Bigg)\leq{0}\label{eq:44}\\
P_3 & -2\alpha{I}+\alpha^2{P_8}<0,\label{eq:45}\\
P_5 & -P_6 <0,\label{eq:46}\\
\hat{\mathcal{M}}_{2}(0) &\leq{0},\ \hat{\mathcal{M}}_{2}(T_2)\leq {0},\label{eq:47} 
\end{align}
where $\mathcal{M}_{2}(\tau)=\left[\begin{array}{c|c}e^{\delta\tau}\Lambda_{1}+P_6 & e^{\delta\tau}\Lambda_{2}\\\noalign{\vspace{0.7ex}}\hline\noalign{\vspace{0.7ex}}
e^{\delta\tau}\Lambda^{\T}_{2} &  -P_3\end{array}\right]$, $\Lambda_{1}={\He}(Z_1)-\delta{P_2}$, $\Lambda_{2}=\begin{bmatrix}Z_1C_0-P_2C_0A_0 & Z_2\end{bmatrix}$, $\Psi=\begin{bmatrix}Y & I\\V^{\T} & 0\end{bmatrix}$, 
\begin{align}
\Pi_{1}&=\begin{bmatrix}A{Y}+BK_1 & A_0+B{K_2}C\\
K_3 & XA+K_{4}C\end{bmatrix},~\Pi_{2}=\begin{bmatrix}BK_2\\K_4\end{bmatrix}.\label{eq:34a2}
\end{align} 
Let $U\in\mathbb{R}^{(n_p+n_z)\times(n_p+n_z)}$ be any nonsingular matrix such that
\begin{equation}\label{eq:48}
XY+UV^{\T}=I.
\end{equation}
Then the conditions in Theorem \ref{th:2} and subsequently those of \textit{Property 1} are satisfied. With the selection of control and hold gains as in 
\begin{align}
\Delta_{\mathcal{K}}&=\begin{bmatrix}U^{-1} & \!\!\!-U^{-1}\!X\!B\\0 & I\end{bmatrix}\!\!\begin{bmatrix}K_3\!-\!XA{Y} & \!\!\!K_4\\K_1 & \!\!\!K_2\end{bmatrix}\!\!\begin{bmatrix}V^{-T} & \!\!\!0\\\!-C{Y}V^{-T} & \!\!\!I\end{bmatrix},\nonumber\\ 
\Delta_{\mathcal{J}}&=\left[\begin{array}{c|c}P^{-1}_{2}Z_1 & P^{-1}_{2}E\end{array}\right],\label{eq:con_hold}
\end{align} 
the solution to \textit{Problem 1} is obtained.
\end{myproposition}
\begin{proof}
We have previously stated that the sufficient stability conditions in Theorem \ref{th:2} satisfy the requirements of \textit{Property 1} and subsequently achieve global exponential stability of $\mathcal{H}_{cl}$. However, the conditions in Theorem \ref{th:2} cannot be immediately adopted to pursue the design of stabilizer and hold control matrices $\Delta_{\mathcal{K}}$, $\Delta_{\mathcal{J}}$ and therefore if we can now show that the equations \eqref{eq:43} - \eqref{eq:47} provide an alternative stability conditions in terms of LMI, then \textit{Problem 1} is turned into a feasibility problem of these LMIs and solution to these LMIs eventually lead us to derive $\Delta_{\mathcal{K}}$, $\Delta_{\mathcal{J}}$ of \textit{Problem 1}.

\textit{proof of \eqref{eq:43}}: Since $P_1>0$ with $P_1$ in \eqref{eq:38} is not an LMI with respect to $X,Y$ and $V$, we perform congruence transformation on $P_1>0$ with multiplying $\Psi=\begin{bmatrix}Y & I\\V^{\T} & 0\end{bmatrix}$ and $\Psi^{\T}$ respectively on either side of this inequality to yield $\bar{P}_{1}>0$ in \eqref{eq:43}. The inequality in \eqref{eq:43} is linear with respect to $X,Y$ and ${P}_{1}>0\Longleftrightarrow{\bar{P}_1}>0$.

\textit{Proof of \eqref{eq:44}}: $\mathcal{M}_{1}$ in \eqref{eq:28} is nonlinear with respect to $P_1,A_c,B_c,C_c,D_c$. By using Schur complement on $\mathcal{M}_{1}$, an equivalent simpler inequality in terms of $\bar{\mathcal{M}}_{1}$ is obtained in \eqref{eq:33}. After congruence transformation on \eqref{eq:33} and subsequent change of variables, $\hat{\mathcal{M}_{1}}$ yields \eqref{eq:44}, which is now linear with respect to $X,Y,V,K_1,K_2,K_3,K_4,P_5,P_8$.

\textit{Proof of \eqref{eq:45}}: As we have noted earlier, $\hat{\mathcal{M}}_{1}$ is linear with respect to $P_8$. However, with the substitution of $P_4=P^{-1}_{8}$, \eqref{eq:26} yields \eqref{eq:34} which is not an LMI with respect to $P_8$. From Lemma 1 of \cite{Ferrante12}, we have shown earlier that the LMI in \eqref{eq:45} indeed satisfies \eqref{eq:35} for any $\alpha\in\mathbb{R}$.

Notice that the LMI conditions in \eqref{eq:27} and \eqref{eq:46} are unchanged. Furthermore, the proof of \eqref{eq:47} directly follows from \eqref{eq:29} by setting $Z_1=P_2{H}$ and $Z_{2}=P_2{E}$. Therefore, the set of LMIs in \eqref{eq:43} - \eqref{eq:48} are equivalent sufficient stability conditions for asserting global exponential stability of $\mathcal{H}_{cl}$. Furthermore, once we solve these LMIs, then with the decision variables in \textit{Proposition 1}, we construct our stabilizer and hold matrix $\Delta_{\mathcal{K}}$ and $\Delta_{\mathcal{J}}$ respectively in \eqref{eq:41} and \eqref{eq:42.1}. This concludes the proof.
\end{proof}
\subsection{Illustrative Example:} \label{sub:sec} In this section, we present a numerical example to illustrate the effectiveness of our designed stabilizer for the following plant:
\begin{equation}
\begin{aligned}
A_p & =\begin{bmatrix}-2 &1\\0 & -0.8\end{bmatrix},B_p=\begin{bmatrix}0\\1\end{bmatrix},E_p\!=\!\begin{bmatrix}1 &0\end{bmatrix},\\
C_p & = 10E_p,~F_p\!=20B^{\T}_{p}, \label{eq:ex1}
\end{aligned}
\end{equation}
and exosystem \eqref{eq:2} is a single frequency harmonic oscillator of the form
\begin{equation}\label{eq:ex2}
S = \begin{bmatrix}0 & 1\\-1 & 0\end{bmatrix},~\sigma(S)\ =\ \pm\ 1i.
\end{equation}
It is easy to verify that the Assumptions \ref{as:1}, \ref{as:2} are satisfied. Then, according to \cite{14}, we select $1$- copy internal model of the exosystem \eqref{eq:int_model} as
\begin{equation}\label{eq:eq36}
G_1=\begin{bmatrix}0 & 1\\-1 & 0\end{bmatrix}, G_2=\begin{bmatrix}0\\1\end{bmatrix}. 
\end{equation} 
With the above system parameters of the extended plant model $\hat{\mathcal{P}}$ in Fig. \ref{fig.1}, we then proceed to obtain a numerical solution to the set of LMIs in Proposition~\ref{prop.1}. Numerical solutions to LMIs are obtained through a YALMIP toolbox \cite{Lofberg7} in Matlab\textsuperscript{\textcopyright} and SDPT3 solver \cite{tutuncu6}. To enforce the nonsingularity of matrix $V$ in \textit{Proposition} \ref{prop.1} and avoid ill-conditioned controller matrix $\Delta_{\mathcal{K}}$ \eqref{eq:con_hold}, following \cite{Ferrante12}, we respectively consider two additional constraints
\begin{align}
V + V^{\T} \succ 0,~ -50\bar{P}_{1} \preceq  {\Pi_1+\Pi^{\T}_{1}} \preceq -0.2\bar{P}_{1}, \label{eq:53}
\end{align}
where matrices $\bar{P}_{1}$ and $\Pi_{1}$ are defined in \eqref{eq:43} and \eqref{eq:34a2}. With our proposed design methodology, we find a feasible solution to the LMIs in \textit{Proposition} \ref{prop.1} for $T_2 = 0.3$, $\alpha = 4.35$, $\delta=3.5$ and the stabilizer and hold gain matrices are given as follows
\begin{align}
A_c & \!=\!\! \begin{bmatrix}\!-2.05 & \!0.78 & \!-0.084 & \!-0.063\\\!-4.76 & \!-3.27 & \!-0.417 & \!-0.055\\
\!-7.799 & \!-6.4 & \!-4.47 & \!1.259\\
\!40.98 & \!19.73 & \!-2.48 & \!-20.32
\end{bmatrix}\!,B_c\!=\!\begin{bmatrix}0.003\\0.13\\0.2\\-1.08\end{bmatrix}\!,\nonumber\\
C_c & = \begin{bmatrix}184.05 & \!\!103.39 & \!\!11.27 & \!\!-78.5\end{bmatrix}\!,D_c\!=\!-4.827,\nonumber\\
H & =-2.135, E=\!\begin{bmatrix}-2.13 & \!\!\!0.0014 & \!\!\!0.0171 & \!\!\!-0.0003\end{bmatrix}\!\!. \nonumber
\end{align}
With these controller parameters, and minimum dwell time $T_1 = 0.1$, we observe from Figs.~\ref{fig.4} that the plant output successfully tracks the reference exosystem trajectory $F_pw$ and the regulated output goes to zero. 
 \begin{figure}
  \centering
  \includegraphics[width=1\linewidth]{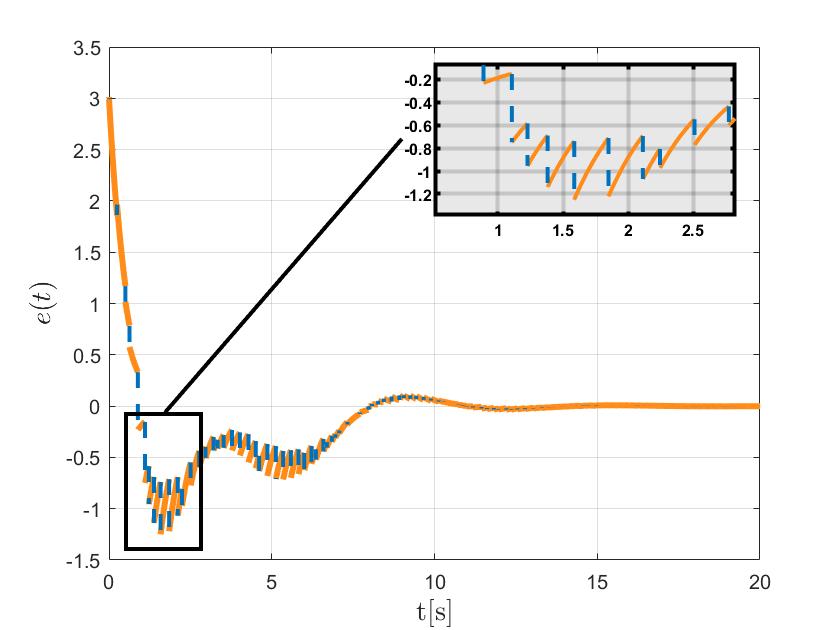} 
\caption{Sample and hold regulated output $e(t)$ with dashed lines indicating jumps in it.}\vspace{-2.5ex}
\label{fig.4}
\end{figure}

\section{Postprocessing with hybrid internal model}\label{sec:7}
\begin{figure}
\centering
\includegraphics[width=0.48\textwidth]{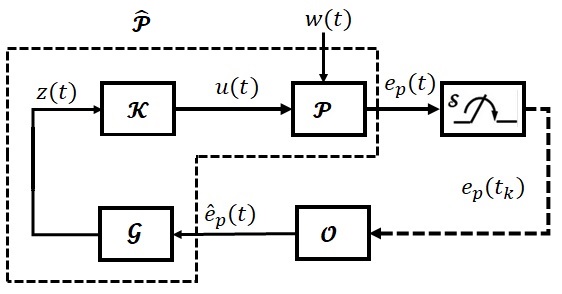}
\caption{{\label{fig.1a}} Schematic representation of the closed loop system with sporadic measurements under ``post-processing'' architecture. Continuous-time signals are marked with solid arrows, while the sporadic measurements are with dashed arrows.}\vspace{-1ex}
\end{figure}
In this section, we illustrate how the postprocessing paradigm can be exploited in the setting  of this paper to solve the output regulation problem. In particular, one interesting aspect that emerges is that in postprocessing internal model control, the internal model is continuously fed via the regulation error, which however in our setting is available only intermittently. To overcome this problem, we propose to augment a classical internal-model $\mathcal{G}$ with a hybrid estimator $\mathcal{O}$, as shown in Fig. \ref{fig.1a}, which provides a converging estimate of the regulation error based on the intermittent measurement. This gives rise to a ``hybrid internal model''. With a little abuse of notation, we consider the following architecture for the internal model: 
\begin{equation}\label{eq:int_model2}
\mathcal{G}\left\{
\begin{aligned}
&\dot{z}=G_{1}z+G_{2}\hat{e}_{p},\hspace{18ex}\text{if}~t\neq{t_k},\\
&z(t^{+})= z(t),\hspace{22ex}\text{if}~t={t_k},\\
&u_{i}= Kz,
\end{aligned}\right.
\end{equation} 
where $z\in\mathbb{R}^{n_z}$ is the state vector of the internal model $\mathcal{G}$, and the input to the internal model $\hat{e}_{p}$ is an estimated output of $e_{p}$ by the hybrid observer $\mathcal{O}$
\begin{equation}\label{eq:h.o}
\mathcal{O}\left\{
\begin{aligned}
&\dot{\chi} = \mathsf{T}\chi,\hspace{26ex}\text{if}~t\neq{t_k}\\
&\chi(t^+) =\mathsf{L_1}\chi(t)+\mathsf{L_2}e_p(t),\hspace{9.5ex}\text{if}~t={t_k},\\
&\hat{e}_{p} = \mathsf{H}\chi,
\end{aligned}\right.
\end{equation}
with the observer matrices $\mathsf{L}_i,~i=1,2$, and $\mathsf{T}$ to be designed later and $\chi\in\R^{n_\chi}$ being the state variable of the hybrid observer $\mathcal{O}$. The output $u_i$ in \eqref{eq:int_model2} of the above internal model is fed to a standard LTI continuous-time stabilizer $\mathcal{K}$ with state vector $x_k\in\mathbb{R}^{n_k}$ as
\begin{equation}\label{eq:40}
\mathcal{K}\left\{\begin{aligned}
&\dot{x}_k = A_kx_k+B_k{u_i},\hspace{12ex}\text{if}~t\neq{t_k},\\
&x_k(t^+) = x_k(t),\hspace{16ex}\text{if}~t={t_k},\\
& u(t) = C_k x_k(t)+D_ku_{i}(t).
\end{aligned}\right.
\end{equation}
Before moving forward, we suppose the following assumption holds.
\begin{myassumptions}
\label{assu:francis}
Let 
$$
A_{cl}=\!\!\begin{bmatrix}
A_p&\!\!\!B_pC_k\\
0&\!\!\!A_k
\end{bmatrix}, B_{cl}=\!\!\begin{bmatrix}
B_pD_k\\
B_k
\end{bmatrix}\!\!K, \mathsf{H}_{1}= \begin{bmatrix}C_p&\!\!\!\!0_{p\times{n_k}}\end{bmatrix},
$$
and for all $\lambda\in\sigma(S)$, the matrix
$$\begin{bmatrix}A_{cl}-\lambda I & B_{cl}\\\mathsf{H}_{1} & 0\end{bmatrix}
$$
is full rank.
\end{myassumptions}

Assumption~\ref{assu:francis} ensures the existence of a pair $(X_M, Z)$ such that: 
\begin{equation}\label{eq:Francis}
\begin{aligned}
X_MS &=A_{cl}X_M+B_{cl}{Z}+E_{cl},\\
0 &=\mathsf{H}_{1}X_M-F_p,
\end{aligned}
\end{equation} 
where $E_{cl}^{\T}\coloneqq \begin{bmatrix}
E_p^{\T} &0
\end{bmatrix}$, $Z\in\mathbb{R}^{q\times{n_z}}$ satisfying $ZS=GZ$ by virtue of internal model principle. We show later that Assumption~\ref{assu:francis} is fulfilled under some natural conditions. Let us define an augmented stabilizer-plant system with the state vector $x_{M}=\text{col}(x_p,x_k)$ as   
\begin{equation}
\begin{aligned}
& \dot{x}_{M}= A_{cl}x_{M}+B_{cl}z+E_{cl}w,\hspace{6ex}\text{if}~t\neq{t_k},\\
& e_{p}= \mathsf{H}_{1}x_{M}-F_{p}w,\\
& x_{M}(t^{+})= x_{M},\hspace{19.8ex}\text{if}~t={t_k},
\end{aligned}
\end{equation}
where internal model state $z$ evolves according to \eqref{eq:int_model2} and is rewritten as
\begin{equation}
\begin{aligned}
&\dot{z}=G_{1}z+G_{2}\mathsf{H}_{1}\chi,\hspace{16ex}\text{if}~t\neq{t_k},\\
&z(t^{+})=z(t),\hspace{21.7ex}\text{if}~t={t_k}.
\end{aligned}
\end{equation}
By using the change of coordinates $\tilde{x}_{M}=x_{M}-X_{M}w, \tilde{z}=z-Zw$ and defining $\tilde{x}_{\alpha}=\text{col}(\tilde{x}_{M},\tilde{z})$ with $X_{M}$ and $Z$ being solution to \eqref{eq:Francis} and \eqref{eq:Z} respectively, the augmented system $\hat{\mathcal{P}}$ in Figure \ref{fig.1a} consisting of the stabilizer-plant and internal model dynamics in transformed coordinates yield the following form
\begin{equation}\label{eq:eq40}
\begin{aligned}
& \dot{\tilde{x}}_{\alpha} = \mathfrak{A}_{c}\tilde{x}_{\alpha}+\mathfrak{B}_{c}(\mathsf{H}\chi-\begin{bmatrix}\mathsf{H}_{1}&0\end{bmatrix}\tilde{x}_{\alpha}),\hspace{2.4ex}\text{if}~t\neq{t_k},\\
& \tilde{x}_{\alpha}(t^{+}) = \tilde{x}_{\alpha},\hspace{24.8ex}\text{if}~t={t_k}, 
\end{aligned}
\end{equation}
where 
\begin{equation}
\begin{aligned}
\mathfrak{A}_{c}= \mathfrak{A}\!+\!\mathfrak{B}_{c}\begin{bmatrix}\mathsf{H}_{1}&0\end{bmatrix}, \mathfrak{A}=\begin{bmatrix}A_{cl} & B_{cl}\\0 & G_{1}\end{bmatrix}, \mathfrak{B}_{c} = \begin{bmatrix}0\\G_{2}\end{bmatrix}.\nonumber
\end{aligned}
\end{equation}
At this stage it is clear that to solve the output regulation problem, $\mathsf{T}$ and $\mathsf{L}_{i},~i=1,2$ in \eqref{eq:h.o} need to be designed to ensure that $\tilde{e}_p\coloneqq (e_p-\hat{e}_p)$ approaches zero asymptotically. To this end, we pick $n_\chi=n_p+n_k+n_z+p$, we denote 
$\chi=(\chi_1, \chi_2)\in\R^{n_p+n_k+n_z}\times \R^{p}$. The state component $\chi_{1}$ can be viewed as an estimate of $\tilde{x}_{\alpha}$ in \eqref{eq:eq40} and $\chi_{2}$ of $\tilde{e}_{p}$. Based on this, let us now select the observer parameters as
\begin{equation}\label{eq:eq45}
\begin{aligned}
\mathsf{T} & = \begin{bmatrix}\mathfrak{A}_{c} & Q\\0 & W\end{bmatrix}, \mathsf{L}_{1}=\begin{bmatrix}I & 0\\-\mathsf{H}_{2} & 0\end{bmatrix}, \mathsf{L}_{2}=\begin{bmatrix}0\\I\end{bmatrix},\\
\mathsf{H}_{2} &=\begin{bmatrix}\mathsf{H}_{1} & 0\end{bmatrix}, \mathsf{H}=\begin{bmatrix}\mathsf{H}_{2}&0\end{bmatrix},
\end{aligned}
\end{equation}
where $Q\in\R^{(n_p+
n_k+n_z)\times p}$ and $W\in\R^{p\times p}$ are to be designed. By introducing a clock variable $\tau$, defining $\tilde{\chi}_{1}=\tilde{x}_{\alpha}-\chi_{1}$, $\tilde{\chi}_{2}=\chi_{2}-\mathsf{H}_{2}\tilde{\chi}_{1}$, the overall closed-loop hybrid dynamical system with state vector $\tilde{\mathbf{x}}=\text{col}(\tilde{x}_{c},\tau), \tilde{x}_{c}=\text{col}(\chi_{1},\tilde{\chi}_{1},\tilde{\chi}_{2})$ yields \eqref{eq:14}, where
\begin{equation}\label{eq:eq42}
\begin{aligned}
f(\tilde{\mathbf{x}}) & = \text{blk diag}\left(\begin{bmatrix}\mathfrak{A}_{c} & \begin{array}{cc}Q\mathsf{H}_{2}& Q\end{array}\\0 & h(\tilde{\chi}_{1},\tilde{\chi}_{2})\end{bmatrix},-1\right),\\
h(\tilde{\chi}_{1},\tilde{\chi}_{2}) & = \begin{bmatrix}\mathfrak{A}-Q\mathsf{H}_{2} & -Q\\W\mathsf{H}_{2}-\mathsf{H}_{2}(\mathfrak{A}-Q\mathsf{H}_{2}) & W+\mathsf{H}_{2}Q\end{bmatrix},\\
\mathcal{C} & = \mathbb{R}^{2(n_p+n_k+n_z)+p}\times\left[0,T_2\right],\\
G(\tilde{\mathbf{x}}^{+}) & =\text{blk diag}(I,0,\tau^{+}), \tau^{+}\in\left[T_1,T_2\right],\\
\mathcal{D} & = \mathbb{R}^{2(n_p+n_k+n_z)+p}\times\{0\}.
\end{aligned}
\end{equation}
Let us now consider the following assumption.
\begin{myassumptions}
\label{as:LyapIM}
The matrix $\mathfrak{A}_c$ is Hurwitz  and there exist $\mathcal{W}\colon\R^{n_\chi+1}\rightarrow\R$ and positive real numbers $\boldsymbol{\omega}_1$, $\boldsymbol{\omega}_2$, $\boldsymbol{\varrho}$ such that for all $(\tilde{\chi}_1, \tilde{\chi}_2,\tau)\in\R^{n_\chi}\times [0, T_2]$:
$$
\begin{aligned}
&\boldsymbol{\omega}_1(\tilde{\chi}_1^{\T} \tilde{\chi}_1+\tilde{\chi}_2^{\T} \tilde{\chi}_2)\leq \mathcal{W}(\tilde{\chi}_1, \tilde{\chi}_2,\tau)\leq\boldsymbol{\omega}_2(\tilde{\chi}_1^{\T} \tilde{\chi}_1+\tilde{\chi}_2^\top \tilde{\chi}_2)\\
&\langle\nabla{\mathcal{W}}(\tilde{\chi}_1, \tilde{\chi}_2,\tau), (h(\tilde{\chi}_1, \tilde{\chi}_2),-1) \rangle\leq -\boldsymbol{\varrho} (\tilde{\chi}_1^\top \tilde{\chi}_1+\tilde{\chi}_2^\top \tilde{\chi}_2)\\
&\mathcal{W}(\tilde{\chi}_1, 0,\nu)-\mathcal{W}(\tilde{\chi}_1, \tilde{\chi}_2,0)\leq 0\qquad \forall\nu\in[T_1, T_2]
\end{aligned}
$$
\end{myassumptions}

The result given next shows that under Assumption~\ref{as:LyapIM} the output regulation problem is solved.
\begin{mytheorem}
\label{the:OutputObserver}
Let Assumption~\ref{as:LyapIM} hold. Then, the following items hold:
\begin{itemize}
    \item[($i$)] Assumption~\ref{assu:francis} holds;
    \item[($ii$)] the set
\begin{equation}
\label{eq:setIM}
\mathcal{A}_I\coloneqq \{0\}\times [0, T_2]\subset \R^{2(n_p+n_z+n_k)+p}\times\R
\end{equation}
is globally exponentially stable for \eqref{eq:14} with flow/jump maps and their respective domains given in \eqref{eq:eq42}.
\end{itemize}
\end{mytheorem}
\begin{proof}
To prove item $(i)$, it suffices to notice that if 
$$\mathfrak{A}_c=\mathfrak{A}+\mathfrak{B}_{c}\mathsf{H}_{2}$$
is Hurwitz, then the triple
$\left(\mathfrak{A},\mathfrak{B}_{c},\mathsf{H}_{2}\right)$
is stabilizable and detectable. Therefore, one has that for all $\lambda\in\mathbb{C}_+$, the matrix 
$$
\begin{aligned}
\rk&\begin{bmatrix}
A_{cl}-\lambda I&B_{cl}&0\\
0&G_1-\lambda I&G_2
\end{bmatrix}=n_p+n_k+n_z\\
\rk&\begin{bmatrix}
A_{cl}-\lambda I&B_{cl}\\
0&G_1-\lambda I\\
\mathsf{H}_{1}&0
\end{bmatrix}=n_p+n_k+n_z
\end{aligned}
$$
From the second condition above, it turns out that for all $\lambda\in\sigma(G_1)=\sigma(S)\subset\mathbb{C}_+$
$$
\rk\begin{bmatrix}
A_{cl}-\lambda I&B_{cl}\\
\mathsf{H}_{1}&0\\
0&G_1-\lambda I
\end{bmatrix}=n_p+n_k+n_z
$$
which is equivalent to Assumption~\ref{assu:francis}.

To show item $(ii)$, we rely on the cascade structure of \eqref{eq:14} with flow/jump maps and domain sets given in \eqref{eq:eq42}. Define for all $\boldsymbol{\tilde{x}}\in\mathcal{C}$
$$
\mathcal{U}(\boldsymbol{\tilde{x}})\coloneqq \chi_1^{\T}\boldsymbol{\mathcal{P}}\chi_1+ \mathcal{W}(\tilde{\chi}_1,\tilde{\chi}_2,\tau)
$$
where $\boldsymbol{\mathcal{P}}=\boldsymbol{\mathcal{P}}^{\T}> 0$ is such that 
$$
\He(\boldsymbol{\mathcal{P}}\mathfrak{A}_c)=-\boldsymbol{\mathcal{Q}}
$$
for some $\boldsymbol{\mathcal{Q}}=\boldsymbol{\mathcal{Q}}^{\T}> 0$. Such a selection of $\boldsymbol{\mathcal{P}}$ and $\boldsymbol{\mathcal{Q}}$ is possible due to $\mathfrak{A}_c$ being Hurwitz.
Observe that for all $\boldsymbol{\tilde{x}}\in\mathcal{C}$,
\begin{align}
& \boldsymbol{\alpha}_1\vert \boldsymbol{\tilde{x}}\vert^2_{\mathcal{A}_{I}} \leq \mathcal{U}(\mathbf{\tilde{x}})   \leq \boldsymbol{\alpha}_2\vert \mathbf{\tilde{x}}\vert^2_{\mathcal{A}_{I}},\text{with}\label{eq:SandwichW}\\
& \boldsymbol{\alpha}_1=\max\{\lambda_{\max}(\boldsymbol{\mathcal{P}}),\boldsymbol{\omega}_2\}, \boldsymbol{\alpha}_2= \min\{\lambda_{\min}(\boldsymbol{\mathcal{P}}),\boldsymbol{\omega}_1\}.\nonumber
\end{align}

Then, from Assumption~\ref{as:LyapIM},
for all $\boldsymbol{\tilde{x}}\in\mathcal{C}$, one has
$$
\begin{aligned}
&\langle \nabla \mathcal{U}(\boldsymbol{\tilde{x}}), f(\boldsymbol{\tilde{x}})\rangle\leq -\chi_1^{\T}\boldsymbol{\mathcal{Q}}\chi_1-\boldsymbol{\varrho}(\tilde{\chi}_1^\top \tilde{\chi}_1+\tilde{\chi}_2^\top \tilde{\chi}_2)\\
&\hspace{6ex}+2\chi_1^{\T}\boldsymbol{\mathcal{P}}Q\mathsf{H}_{2}\tilde{\chi}_1+2\chi_1^{\T}\boldsymbol{\mathcal{P}}Q\tilde{\chi}_2, 
\end{aligned}
$$
which by using Young inequality yields, for all $\boldsymbol{\tilde{x}}\in\mathcal{C}$:
\begin{align}
&\langle \nabla \mathcal{U}(\boldsymbol{\tilde{x}}), f(\boldsymbol{\tilde{x}})\rangle\!\leq\!\chi_1^\top(-\boldsymbol{\mathcal{Q}}+(\varpi_{0}+\varpi_{1}) I)\chi_1+\!\tilde{\chi_1}^{\T}\!\!\left(-\boldsymbol{\varrho}I\right.\nonumber\\
&\left.+\dfrac{1}{\varpi_{0}}\mathsf{H}^{\T}_{2} Q^{\T}\boldsymbol{\mathcal{P}}^2Q\mathsf{H}_{2}\!\!\right)\!\tilde{\chi}_{1}\!+\!\tilde{\chi}^{\T}_{2}\left(\!\!-\boldsymbol{\varrho}I\!\!+\!\!\dfrac{1}{\varpi_{1}}Q^{\T}\boldsymbol{\mathcal{P}}^{2}Q\!\right)\!\tilde{\chi}_{2},\label{eq:Young}
\end{align}
for any $\varpi_{0},\varpi_{1}>0$. At this stage, select $\varpi_{0}$ and $\varpi_{1}$ such that
\begin{equation}\label{eq:eq46}
\begin{aligned}
&-\boldsymbol{\mathcal{Q}}+(\varpi_{0}+\varpi_{1}) I< 0,\\
&-\boldsymbol{\varrho}I+\dfrac{1}{\varpi_{0}}\mathsf{H}^{\T}_{2}Q^{\T}\boldsymbol{\mathcal{P}}^{2}Q\mathsf{H}_{2}~<~0,\\
&-\boldsymbol{\varrho}I+\dfrac{1}{\varpi_{1}}Q^{\T}\boldsymbol{\mathcal{P}}^{2}Q~<~0.
\end{aligned}
\end{equation}
Then, from \eqref{eq:Young} for all $\boldsymbol{\tilde{x}}\in\mathcal{C}$, $\langle \nabla \mathcal{U}(\boldsymbol{\tilde{x}}), f(\boldsymbol{\tilde{x}})\rangle$ yields
\begin{equation}
\begin{aligned}
&\langle \nabla \mathcal{U}(\boldsymbol{\tilde{x}}), f(\boldsymbol{\tilde{x}})\rangle\leq-\lambda_c\vert \boldsymbol{\tilde{x}}\vert^2_{\mathcal{A}_{I}},
\end{aligned}
\label{eq:WdotFinal}
\end{equation}
where $\lambda_c\coloneqq\min\left\{\vert\lambda_{c_1}\vert,\left\vert\lambda_{c_2}\right\vert,\left\vert\lambda_{c_3}\right\vert
\right\}$ with 
\begin{equation}
\begin{aligned}
\lambda_{c_1}&=\lambda_{\max}\left(-\boldsymbol{\mathcal{Q}}+(\varpi_{0}+\varpi_{1}) I\right),\\
\lambda_{c_2}&=\lambda_{\max}\left(-\boldsymbol{\varrho}I+\dfrac{1}{\varpi_{0}}\mathsf{H}^{\T}_{2}Q^{\T}\boldsymbol{\mathcal{P}}^{2}Q\mathsf{H}_{2}\right),\\ 
\lambda_{c_3}&=\lambda_{\max}\left(-\boldsymbol{\varrho}I+\dfrac{1}{\varpi_{1}}Q^{\T}\boldsymbol{\mathcal{P}}^{2}Q\right).
\end{aligned}
\end{equation} 
For a specific choice of $\varpi_{0}$ and $\varpi_{1}$ satisfying $\dfrac{\varpi_{0}}{\varpi_{1}}\geq\|\mathsf{H}_{2}\|^2$, $\varpi_{0}+\varpi_{1}<\lambda_{\min}(\boldsymbol{\mathcal{Q}})$, it can be easily shown that all the conditions in \eqref{eq:eq46} hold. Now observe that, from Assumption~\ref{as:LyapIM}, for all $\boldsymbol{\tilde{x}}\in\mathcal{D}$, $g\coloneqq (g_{\chi_1}, g_{\tilde{\chi_1}},g_{\tilde{\chi_2}},g_{\tau})\in G(\boldsymbol{\tilde{x}})$ one has 
\begin{equation}
\label{eq:DeltaW}
\mathcal{U}(g)-\mathcal{U}(\boldsymbol{\tilde{x}})=\mathcal{W}(g_{\tilde{\chi_1}},g_{\tilde{\chi_2}},g_{\tau})-\mathcal{W}(\tilde{\chi}_1,\tilde{\chi}_2,0)\leq 0.
\end{equation}
Let $\phi$ be any maximal solution to \eqref{eq:14}. Then, by using \eqref{eq:WdotFinal} and  \eqref{eq:DeltaW}, direct integration of $\mathcal{U}\circ\phi$ yields:
$$
\mathcal{U}(\phi(t, j))\leq e^{-\lambda_c t}\mathcal{U}(\phi(0,0))\quad\forall (t,j)\in\dom\phi
$$
The latter, thanks to Lemma~1 of \cite{Ferrante3} yields, for some positive, solution independent, $\lambda_{0},\kappa$:
$$
\mathcal{U}(\phi(t, j))\leq \kappa e^{-\lambda_{0} (t+j)}\mathcal{U}(\phi(0,0))\quad\forall (t,j)\in\dom\phi
$$
which, thanks to \eqref{eq:SandwichW}, shows that the set $\mathcal{A}_I$ in \eqref{eq:setIM} is GES for \eqref{eq:14}. This concludes the proof.
\end{proof}
\begin{myproposition}
Let $\bar{P}\in\mathbb{R}^{n_p+n_k+n_z}>0$, $\hat{P}\in\mathbb{R}^{p}>0$, $\bar{J}\in\R^{(n_p+n_k+n_z)\times p}$, $\hat{J}\in\R^{p\times p}$, $\delta>0$, such that
\begin{equation}
    \label{eq:LMI_Design_Obser}
    \begin{aligned}
    &\begin{bmatrix}
    \He(\bar{P}\mathfrak{A}-\bar{J}\mathsf{H}_{2}) & -\bar{J}+(-\hat{P}\mathsf{H}_{2}\mathfrak{A}+\hat{J}\mathsf{H}_{2})^{\T}\\
    \bullet&-\delta \hat{P}+\He(\hat{J})
    \end{bmatrix}< 0\\
    &\begin{bmatrix}
   \He(\bar{P}\mathfrak{A}\!-\!\bar{J}\mathsf{H}_{2})&\!\!\!\!-\bar{J}+e^{\delta{T}_{2}}\left[-\hat{P}\mathsf{H}_{2}\mathfrak{A}+\hat{J}\mathsf{H}_{2}\right]^{\T}\\
    \bullet&\!\!\!\!e^{\delta{T}_{2}}\left[-\delta\hat{P}+\He(\hat{J})\right]\!
    \end{bmatrix}\!\!<\!0
    \end{aligned}
\end{equation}
and $\mathfrak{A}_c$ be Hurwitz. Set $Q=\bar{P}^{-1}\bar{J}$, $W=\hat{P}^{-1}\hat{J}-\mathsf{H}_{2}Q$. Then, Assumption~\ref{as:LyapIM} holds with:
\begin{equation}\label{eq:eq51}
\mathcal{W}(\tilde{\chi}_1,\tilde{\chi}_2,\tau)=\tilde{\chi}_1^{\T} \bar{P}\tilde{\chi}_1+e^{\delta\tau} \tilde{\chi}_2^{\T}\hat{P}\tilde{\chi_2}
\end{equation}
\end{myproposition}
\begin{proof}
It is easy to show from \eqref{eq:eq51} that for all $(\tilde{\chi},\tau)\in\R^{n_\chi}\times [0, T_2]$, $\tilde{\chi}=\text{col}(\tilde{\chi}_{1},\tilde{\chi}_{2})$, 
\begin{align}
&\boldsymbol{\omega}_{1}\|\tilde{\chi}\|^{2}\leq\mathcal{W}(\tilde{\chi},\tau)\leq\boldsymbol{\omega}_{2}\|\tilde{\chi}\|^{2}~\text{with}\nonumber\\
&\boldsymbol{\omega}_{1}=\min\left(\lambda_{\min}(\bar{P},\lambda_{\min}(\hat{P})\right),\nonumber\\&\boldsymbol{\omega}_{2}=\max\left(\lambda_{\max}(\bar{P},e^{\delta{T}_{2}}\lambda_{\max}(\hat{P})\right)\nonumber
\end{align}
and thus the first condition in Assumption \ref{as:LyapIM} is satisfied. By differentiating $\mathcal{W}$ in \eqref{eq:eq51} along the flow map $h(\tilde{\chi})$ in \eqref{eq:eq42} and with the substitution of $Q=\bar{P}^{-1}\bar{J}$ and $W=\hat{P}^{-1}\hat{J}-\mathsf{H}_{2}Q$ in the resulting expression, we obtain
\begin{align}
&\dot{\mathcal{W}}(\tilde{\chi},\tau)=\langle\nabla{\mathcal{W}}(\tilde{\chi},\tau), (h(\tilde{\chi}),-1) \rangle=\tilde{\chi}^{\T}\mathcal{M}(\tau)\tilde{\chi},\\
&\mathcal{M}(\tau)\!=\!\!\!\begin{bmatrix}
   \!\He(\bar{P}\mathfrak{A}\!-\!\bar{J}\mathsf{H}_{2})&\!\!-\bar{J}\!\!+\!e^{\delta\tau}\!\!\left(-\mathfrak{A}^{\T}\mathsf{H}^{\T}_{2}\hat{P}\!+\!\mathsf{H}_{2}^{\T}\hat{J}\right)\!\\
    \bullet&\!\!\!\!e^{\delta\tau}\!\!\left(-\delta\hat{P}\!+\!\He(\hat{J})\right)\!\!\!\!
    \end{bmatrix}\!\!.\nonumber
\end{align}
From Lemma 3 of \cite{Ferrante3}, it is straightforward to show that there exists $\lambda^{'}:[0,\tau]\to\left[0,1\right]$ such that for each $\tau\in[0,T_2]$, $\mathcal{M}(\tau)=\lambda^{'}(\tau)\mathcal{M}(0)+(1-\lambda^{'}(\tau))\mathcal{M}(T_2)$. The satisfaction of the two LMIs in \eqref{eq:LMI_Design_Obser} then yields $\mathcal{M}(\tau)<0, \forall{\tau}\in\left[0,T_2\right]$, and thus the second condition in Assumption \ref{as:LyapIM} is satisfied. Furthermore, $\mathcal{W}(\tilde{\chi}_{1},0,v)-\mathcal{W}(\tilde{\chi},0)={-\tilde{\chi}_{2}}\hat{P}\tilde{\chi}_{2}<0$ and the third condition in Assumption \ref{as:LyapIM} is satisfied as well. From Theorem \ref{the:OutputObserver}, the set $\mathcal{A}_{I}$ in \eqref{eq:setIM} of the hybrid closed-loop system \eqref{eq:14} with flow/jump dynamics \eqref{eq:eq42} is therefore globally exponentially stable. This concludes the proof.
\end{proof}

To illustrate the effectiveness of the post-processing paradigm, let us reconsider the same example from Section \ref{sub:sec}. We design a continuous-time stabilizer satisfying \eqref{eq:40} with the stabilizer matrices given as follows
\begin{align}
A_{k} & = \begin{bmatrix}-7 &\!\!-3.28 &\!\!-1.03 &\!\! -1.18 &\!\!0.7\\
 0 &\!\!-9.7 &\!\!12.1 &\!\!0 &\!\!0\\
0 &\!\!-12.1 &\!\!-5.78 &\!\!-4.19 &\!\!2.49\\
0 &\!\!0 &\!\! 0 &\!\!-16.6 &\!\!6.8\\
0 &\!\!0 &\!\!0 &\!\!-6.8\!\!&\!\!-16.6
\end{bmatrix},D_k=0,\nonumber\\
B_k & = \begin{bmatrix}-0.2 & -3 & 1 & -2 & 354.9\end{bmatrix}^{\T},\label{eq:57}\\
C_{k} & = \begin{bmatrix}52.73& -178.33 & -56.3 & -64.3 & 38.2\end{bmatrix},\nonumber
\end{align}
which makes the closed-loop system matrix $\mathfrak{A}_{c}$ in \eqref{eq:eq40} Hurwitz. From \eqref{eq:57}, and by solving the LMIs in \eqref{eq:LMI_Design_Obser} we determine hybrid observer components \eqref{eq:eq45} as
$$Q=\left[Q_1, Q_2\right]^{\T}\ \text{with}~ Q_1=\begin{bmatrix}7.43 &\!\!7.44&\!\!-0.0065&\!\!0.01\end{bmatrix},$$ $$Q_2=\begin{bmatrix}0.16 & 2.81 & 5.07 & 0.28 & -0.25\end{bmatrix},~W=1.357.$$
The internal model controller parameter $G_1$ is same as in \eqref{eq:eq36} and $G_2=\begin{bmatrix}-5&-4\end{bmatrix}^{\T}$. The post-processing regulator comprising of this hybrid observer, stabilizer and internal model controller parameters, given above, is then applied to the plant \eqref{eq:plant} and the simulation results are shown in Figure \ref{fig.ep}. We observe from Fig. \ref{fig.ep} that the hybrid observer successfully estimates the regulated output and along with the actual regulated output of the plant, this also converges to zero asymptotically. Thus, the objective of the output regulation problem is achieved under sporadic measurements.
\begin{figure}
  \centering
  \includegraphics[width=1\linewidth]{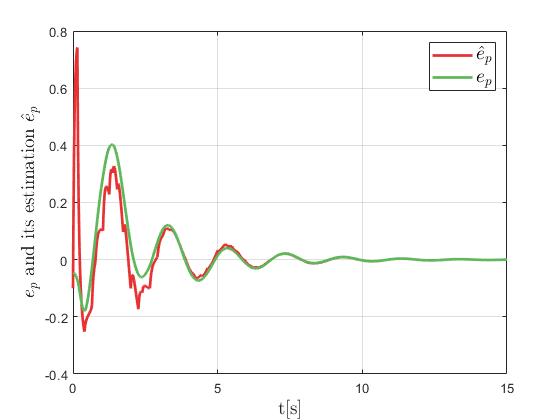} 
\caption{Regulated output of the plant and estimated regulated output by the observer.}\vspace{-3ex}
\label{fig.ep}
\end{figure}

\begin{myremark}
The main advantage introduced by the post-processing approach is that it permits to use a controller designed in a purely continuous-time setting, by extending the closed-loop system (as depicted in Figure~\ref{fig.1a}) via the proposed hybrid observer $\mathcal{O}$, to solve the output regulation problem in the presence of intermittent measurements. This, besides enabling the use of legacy controllers, it also allows to drastically simply the design of the regulator. As opposed, the approach presented in Section~\ref{sec:3}, which is based on the pre-processing paradigm, requires the design of a hybrid controller for the stabilization of the extended plant plant--internal model. This generally results into higher order controllers and a more intricate design approach.     
\end{myremark}

\section{Conclusion}\label{sec:6}
In this paper, we studied the output regulation problem for LTI plants with sporadically sampled output measurements. Due to measurements being available at sporadic time instances, we propose a hybrid regulator that is composed of a continuous-time controller fed by hybrid intersample device.  The resulting closed-loop system is modeled as hybrid dynamical system. Within this setting, we revisit two classical paradigms for output regulation problem: postprocessing and preprocessing internal model. 

In both scenarios, hinging upon Lyapunov tools for hybrid dynamical system, sufficient conditions in the form of matrix inequalities have been provided for the design of a regulator ensuring internal stability and  exponential output regulation. With the use of a continuous-time controller, We also observed that the post-processing paradigm allows greater design simplicity as compared to the pre-processing counterpart. Building upon this post-processing paradigm, in future we would like to extend this study for solving cooperative output regulation problem of multi-agent system under sporadically available measurements from the agents.
\bibliographystyle{IEEEtran}
\bibliography{ref}

\end{document}